\documentclass[11pt]{article}
\usepackage[utf8]{inputenc}
\usepackage[T1]{fontenc}
\usepackage[top=1in, bottom=1in, left=1in, right=1in, letterpaper]{geometry}

\usepackage{amsmath, amssymb, amsfonts, bm, cite}
\usepackage{amsthm}

\usepackage{enumitem}
\usepackage{comment}
\usepackage{xspace}
\usepackage{ifthen}
\usepackage{boxedminipage}
\usepackage{algorithm}
\usepackage{algpseudocode}
\usepackage{color}
\usepackage{empheq}
\usepackage[colorlinks=true,citecolor=blue,bookmarksnumbered=true,hyperfootnotes=true]{hyperref}
\usepackage{cleveref}

\algrenewcomment[1]{\hfill{\scriptsize//~#1}}

\newtheorem{theorem}{Theorem}
\newtheorem{lemma}[theorem]{Lemma}

\newtheorem{corollary}[theorem]{Corollary}

\theoremstyle{remark}

\theoremstyle{definition}
\newtheorem{definition}[theorem]{Definition}

\crefname{theorem}{Theorem}{Theorems}
\crefname{lemma}{Lemma}{Lemmas}
\crefname{proposition}{Proposition}{Propositions}
\crefname{corollary}{Corollary}{Corollaries}
\crefname{fact}{Fact}{Facts}
\crefname{definition}{Definition}{Definitions}
\crefname{remark}{Remark}{Remarks}

\crefname{section}{Section}{Sections}
\crefname{appendix}{Appendix}{Appendices}
\crefname{algorithm}{Algorithm}{Algorithms}

\newcommand{\partdiff}[2]{\frac{\partial {#1}}{\partial {#2}}}

\newcommand{\RR}{{\mathbb R}}

\newcommand{\cI}{\mathcal{I}}
\newcommand{\cM}{\mathcal{M}}

\newcommand{\cH}{\mathcal{H}}

\newcommand{\cA}{\mathcal{A}}
\newcommand{\cB}{\mathcal{B}}
\newcommand{\cG}{\mathcal{G}}
\newcommand{\cP}{\mathcal{P}}
\newcommand{\cZ}{\mathcal{Z}}

\newcommand{\E}{{\mathbb E}}

\newcommand{\vect}[1]{\mathbf{#1}}
\newcommand{\bx}{\vect{x}}

\newcommand{\by}{\vect{y}}
\newcommand{\bz}{\vect{z}}

\newcommand{\bd}{\vect{d}}

\newcommand{\bE}{\vect{E}}

\newcommand{\NSW}{\mbox{NSW}}
\newcommand{\OPT}{\mbox{OPT}}
\newcommand{\Var}{\mbox{Var}}

\DeclareMathOperator*{\argmax}{argmax}

\def\b1{{\bf 1}}

\title{A constant-factor approximation algorithm \\ for Nash Social Welfare with submodular valuations}

\author{Wenzheng Li\thanks{Stanford University} \and Jan Vondr\'{a}k\thanks{Stanford University}}

\begin{document} 
\maketitle \thispagestyle{empty} 

\begin{abstract}
We present a $380$-approximation algorithm for the Nash Social Welfare problem with submodular valuations. Our algorithm builds on and extends a recent constant-factor approximation for Rado valuations \cite{garg2020approximating}.  

\end{abstract}


\section{Introduction}

Nash Social Welfare is the following optimization problem.

\paragraph{Nash Social Welfare (NSW).}
Given $m$ indivisible items and $n$ agents with valuation functions $v_i:2^{[m]} \rightarrow \RR_+$,
we want to allocate items to the agents, that is find a partition of the $m$ items $(S_1,S_2,\ldots,S_n)$ that maximizes the {\em geometric average} of the valuations,
$$ \NSW(S_1,S_2,\ldots,S_n) = \left( \prod_{i=1}^{n} v_i(S_i) \right)^{1/n}.$$
Among the possible objectives considered in allocation of indivisible goods, it can be viewed as a compromise between Maximum Social Welfare (maximizing the summation $\sum_{i=1}^{n} v_i(S_i)$, which does not take fairness into account), and Max-Min Welfare (maximizing $\min_{1 \leq i \leq n} v_i(S_i)$, which focuses solely on the least satisfied agent and ignores the possible additional benefits to others). The notion of Nash Social Welfare goes back to John Nash's work \cite{Nash50} on bargaining in the 1950s. It also came up independently in the context of competitive equilibria with equal incomes \cite{Var74} and proportional fairness in networking \cite{Kel97}. An interesting feature of Nash Social Welfare is that the problem is invariant under scaling of the valuations $v_i$ by independent factors $\lambda_i$; i.e., each agent can express their preference in a ``different currency" and this does not affect the problem.

The difficulty of the problem naturally depends on what class of valuations $v_i$ we consider. Unlike the (additive) Social Welfare Maximization problem, the Nash Social Welfare problem is non-trivial even in the case where the $v_i$'s are {\em additive}, that is $v_i(S) = \sum_{j \in S} v_{ij}$ where $v_{ij}$ is agent $i$'s valuation item $j$. It is NP-hard in the case of 2 agents with identical additive valuations (by a reduction from the Subset Sum problem), and APX-hard for multiple agents \cite{Lee17}. A constant-factor approximation for the additive case was discovered in a remarkable work by Cole and Gatskelis \cite{CG18}, and subsequently via a very different algorithm by Anari et al. \cite{AGSS17}. The algorithm of \cite{CG18} is based on consideration of market equilibria and market-clearing prices. The algorithm of \cite{AGSS17} uses a convex relaxation inspired by Gurvits's work on the permanent of doubly stochastic matrices, which relies on properties of real stable polynomials.
Inspired by these exciting breakthroughs, a series of follow-up work has been developed along these two lines \cite{cole2017convex,barman2018finding,CCG18,AMGV18,GKK20}. The best approximation factor for additive valuations currently stands at $e^{1/e} \simeq 1.45$ \cite{barman2018finding}.

A particularly compelling question is whether a constant-factor approximation is possible for submodular valuations (where a $(1-\frac{1}{e})$-approximation is known for (additive) social welfare maximization \cite{FNW78,Von08}, and submodular valuations are the largest natural class for which such a result is known, assuming only value-oracle access to the valuations). Some progress has been made for Nash Social Welfare with valuations beyond additive ones: a constant factor for concave piece-wise linear separable utilities \cite{AMGV18}, and for budget-additive valuations \cite{GHM19,CCG18}; in fact the approximation factor for budget-additive valuations now matches the $e^{1/e}$ for additive valuations \cite{CCG18}. Recently, \cite{li2021estimating} designed an algorithm to estimate the optimal value within a factor of $\frac{e^3}{(e-1)^2} \simeq 6.8$ for certain subclasses of submodular valuations, such as coverage and summations of matroid rank functions, by extending the techniques of \cite{AGSS17} using stable polynomials.
And most recently, \cite{garg2020approximating} designed a constant-factor $(772)$ approximation algorithm for the class of ``Rado valuations'', which includes matroid rank functions and more generally valuations defined by a certain matching problem with a matroid constraint. \cite{garg2020approximating} presents another significantly different approach to the problem: Instead of  market/pricing-inspired techniques or techniques based on stable polynomials, this paper uses a combination of combinatorial matching techniques and a convex programming relaxation.

For general submodular valuations, the best result prior to this work was an $O(n)$-approximation which also applies to subadditive valuations \cite{GKK20,BBKS20}. However, for subadditive or even fractionally subadditive valuations we cannot expect a constant factor in the value oracle model \cite{BBKS20}, for the same reasons that this is impossible for the Social Welfare Maximization problem \cite{DNS10}. In the special case of a constant number of agents $n$ with submodular valuations, \cite{GKK20} presents a $(1-1/e-\epsilon)$-approximation for any $\epsilon>0$; this algorithm uses an extensive enumeration which makes the running time exponential in $n$.

\paragraph{Our result and techniques.}

\begin{theorem}[Main Result]
\label{thm:main}
There exists a polynomial-time constant-factor approximation algorithm for the Nash Social Welfare problem with monotone submodular valuation functions, accessible by value queries.
\end{theorem}

The approximation factor that we obtain is $380$. We made only modest effort to optimize the constant. We believe that the best constant achievable with the techniques of this paper would still be a triple-digit number. 

Our techniques can be viewed as a natural extension of the approach in \cite{garg2020approximating}. In hindsight, the strength of the approach of \cite{garg2020approximating} is that it is rather modular and isolates the issue of providing at least some nonzero value to each agent as a separate matching problem. The question then remains how to deal with the remaining items and for this we develop some new techniques. The approach of  \cite{garg2020approximating} relies on the existence of a tractable Eisenberg-Gale relaxation with useful polyhedral properties for Rado valuations; this approach might be possibly extended to gross substitutes valuations, but probably not beyond that. The main new components that we introduce are: (i) a new non-convex relaxation of the problem (the Mixed Multilinear Relaxation), (ii) an algorithm to solve it approximately, and (iii) a randomized rounding technique using concentration of submodular functions to obtain an integer solution. We present a more detailed overview at the beginning of Section~\ref{sec:overview}.

\section{Preliminaries}

\paragraph{Nash Social Welfare (NSW).}
Given a set of $m$ indivisible items $\cG$ and a set of $n$ agents $\cA$, with valuation functions 
$v_i:2^{\cG} \rightarrow \RR_+$ for each $i \in\cA$,
we want to allocate the items to the agents, that is find a partition $(S_1,S_2,\ldots,S_n)$ of $\cG$ in order to maximize the {\em geometric average} of the valuations,
$$ \NSW(S_1,\ldots,S_n) = \left( \prod_{i=1}^{n} v_i(S_i) \right)^{1/n}.$$

\paragraph{Monotone Submodular Functions.}
Let $\cG$ be a finite ground set and $v: 2^\cG \to \RR$. 
\begin{itemize}
\item $v$ is submodular if for any $S, T \subseteq  \cG$,  
$$ v(S) + v(T) \ge v(S\cap T) + v(S\cup T).$$ 
\item $v$ is monotone if $v(S)\le v(T)$ whenever $S\subseteq T$.
\end{itemize}

\paragraph{Multilinear Extension} 
For a set function $v: 2^\cG \to \RR$, we define its multilinear extension $V: [0,1]^\cG \to\RR$ by 
$$ V(\bx) = \sum_{S\subset \cG} v(S) \prod_{i\in S} x_i \prod_{j\in \cG\backslash S}(1-x_j).$$

The following is well-known and used in prior work (e.g., \cite{CCPV11}).

\begin{lemma}
Let $V: [0,1]^\cG \to \RR$ be the multilinear extension of a set function $v: 2^\cG \to \RR$.  Then
\begin{itemize}
\item If $v$ is monotone non-decreasing, then $V$ is non-decreasing along any line with direction $\bd\ge 0$.
\item If $v$ is submodular, then $V$ is concave along any line with direction $\bd\ge 0$.
\end{itemize}
\end{lemma}

We use the following shorthand notation: For a singleton set $\{j\}$, we write $v(j)$ to denote $v(\{j\})$. For a set $S$ (either containing or not containing $j$), we write $v(S+j)$ to denote $v(S \cup \{j\})$ and $v(S-j)$ to denote $v(S \setminus \{j\})$. We denote by $\b1_S$ the indicator vector of $S$, i.e. $(\b1_S)_j = 1$ if $j \in S$ and $0$ otherwise. We also write $\b1_j$ instead of $\b1_{\{j\}}$ to simplify the notation.

\section{Our algorithm and analysis}
\label{sec:overview}

\begin{algorithm}[htbp]
\caption{\bf Nash Social Welfare algorithm}
\label{alg:NSW-template}
\begin{algorithmic}[1]
  \Procedure{NSW}{$\cA,\cG,v_1,\ldots,v_n$}:
    \State Find a matching $\tau:\cA \rightarrow \cG$ maximizing $\prod_{i \in \cA} v_i(\tau(i))$
    \State $\cH := \tau(\cA)$, $\cG' := \cG \setminus \cH$, $\cA' := \{i \in \cA: v_i(\cG') > 0\}$
    \State $\by := $ {\bf IteratedContinuousGreedy}$(\cA',\cG',v_1,\ldots,v_n)$
    \State $(R_1,\ldots,R_n) := $ {\bf RandomizedRounding}$(\by)$
    \State Find a matching $\sigma:\cA \rightarrow \cH$ maximizing $\prod_{i \in \cA} v_i(R_i + \sigma(i))$
    \State Return $(R_1 + \sigma(1), R_2 + \sigma(2), \ldots, R_n + \sigma(n))$
  \EndProcedure
\end{algorithmic}
\end{algorithm}

Our algorithm at a high level is described in Algorithm~\ref{alg:NSW-template}.
We are strongly inspired by the algorithm of \cite{garg2020approximating} for Rado valuations and follow their high-level structure. We preserve some of the components of their algorithm but replace components which previously relied on special properties of Rado valuations. The new components are: a new relaxation of the Nash Social Welfare problem, and the subroutines {\bf IteratedContinuousGreedy} and {\bf RandomizedRouding}, which are described and analyzed in Sections~\ref{sec:iter-greedy} and \ref{sec:rand-rounding}, respectively.
The analysis can be summarized as follows (with a numbering of phases analogous to \cite{garg2020approximating}).

\paragraph{Phase I: Initial Matching.}
We find an optimal assignment of 1 item for each agent, i.e.~a matching $\tau:\cA \rightarrow \cG$ maximizing $\prod_{i \in \cA} v_i(\tau(i))$. This is also the starting point in \cite{garg2020approximating}. $\cH$ denotes the items allocated in this matching.

\paragraph{Phase II: Mixed Multilinear Relaxation.}
We formulate an optimization program which aims to assign the items in $\cH$ integrally and the remaining items fractionally under a certain relaxed objective. However, we do not have a concave relaxation at our disposal, such as the Eisenberg-Gale program in \cite{garg2020approximating}; no such tractable relaxation is known for general submodular functions. Instead, we propose a new relaxation involving a product of multilinear functions. 

\begin{align*}
\tag{Mixed-Multilinear}
\label{eqns:MM}
\max \quad & \prod_{i \in \cA} V_i(\bx_i) \\
s.t. \quad & \sum_{i\in\cA} x_{ij} \le 1 && \forall j\in \cG \\
& x_{ij} \ge 0 &&  \\
& x_{ij} \in \{0,1\} && \forall i \in \cA, j \in \cH
\end{align*}

Here, $V_i(\bx_i) = \sum_{S \subseteq \cG} v_i(S) \prod_{j \in S} x_{ij} \prod_{j' \in \cG \setminus S} (1-x_{ij'})$ is the multilinear extension of $v_i$. 

\medskip

Although the items in $\cH$ could be allocated arbitrarily, we will use a matching in the end. Similarly to \cite{garg2020approximating}, we prove that this does not hurt the solution significantly. In the next phase, we deal with the question of solving the fractional part of the relaxation.

\paragraph{Phase III: Iterated Continuous Greedy Algorithm.}
We ignore the items in $\cH$ for a moment and try to solve the optimization problem restricted to the item set $\cG' = \cG \setminus \cH$ and the subset of agents $\cA'$ who have positive value for these items. 

\begin{align*}
\tag{MultilinearProduct}
\label{eqns:MM}
\max \quad & \prod_{i \in \cA'} V_i(\by_i) \\
s.t. \quad & \sum_{i\in\cA'} y_{ij}\le 1 && \forall j\in \cG' \\
& y_{ij} \ge 0 &&  \\
\end{align*}

A natural idea is to apply the continuous greedy algorithm of \cite{CCPV11}. However, a direct application doesn't work since the objective function is not concave even in nonnegative directions (a product of concave functions is not necessarily concave). We can obtain an objective function concave in nonnegative directions, if we take a logarithm of the objective function: The logarithm of a non-decreasing concave function is non-decreasing concave, and we get a summation instead of a product.

\begin{align*}
\tag{LogMultilinear}
\label{eqns:LM}
\max \quad & \sum_{i \in \cA'} \log V_i(\by_i) \\
s.t. \quad & \sum_{i\in\cA'} y_{ij}\le 1 && \forall j\in \cG' \\
& y_{ij} \ge 0 &&  \\
\end{align*}

Nevertheless, the continuous greedy algorithm still doesn't work as such, because it gives a multiplicative approximation; but we require an additive approximation on the logarithmic scale. 

Our solution is an iterated version of the continuous greedy algorithm, where we run the continuous greedy algorithm, scale the solution by a factor of $1/2$, and repeat as long as there is some tangible gain. The intuition is that as long as our solution has low value, the continuous greedy process makes progress at a high rate and hence we gain more in the continuous greedy process than what we lose in the scaling step. 
The output of the iterated continuous greedy algorithm is a solution $\by$ satisfying
$$ \sum_{i \in \cA'} \frac{V_i(\by^*_i)}{V_i(\by_i)} = O(n) $$
where $\by^*$ is the optimal solution. This is a stronger guarantee than just approximating the optimum of (\ref{eqns:LM})
which will be useful in the analysis.

\paragraph{Phase IV: Randomized Rounding.}
Our next goal is to round or at least sparsify the fractional solution $\by$. Since our relaxation doesn't have polyhedral properties which were used for sparsification in \cite{garg2020approximating}, we resort to a more elementary approach:   randomized rounding. We simply allocate each item $j$ to agent $i$ with probability $y_{ij}$. 

Ideally, we would like to argue that the contribution to each agent is strongly concentrated, and thus the value of the assignment is close to the value of the fractional solution. It is known that submodular functions satisfy concentration bounds which can be useful here; the only problem is that the concentration bounds work well only for items with small contributions. 

Hence, we partition the items for each agent into ``large'' and ``small'': Large items are defined greedily by choosing the maximum marginal profit, as long as the total fractional mass of large items does not exceed some constant $c>0$. 
In the analysis, we apply randomized rounding only to the small items. Since their marginal contributions are bounded, we can apply the Efron-Stein inequality and prove that we lose only a constant factor by rounding the small items. The result is a sparsified fractional solution, where only large items are assigned fractionally and their total fractional mass is bounded for each agent.

\paragraph{Phase V: Matching Recombination.}
The last piece of the puzzle is what to do with large items. Luckily, \cite{garg2020approximating} contains a component which is useful exactly for this purpose. 
A key lemma in \cite{garg2020approximating} 
shows that for any fractional solution $\by$ and any matching $\pi:\cA \rightarrow \cH$ (imagine the optimal matching on top of $\by$), there is another matching $\rho: \cA \rightarrow \cH$ such that the value of $(\by, \rho)$ is comparable to the value of $(\by,\pi)$, and for each agent, either the item matched in $\rho$ has a significant value, or there is no item outside of $\cH$ which has a significant value. The matching is obtained by an alternating-cycle procedure applied to the matching $\pi$ and the initial matching $\tau$.

We adapt this lemma and apply it in our setting: After switching to the matching $\rho$, either the matching item $\rho(i)$ itself provides a constant fraction of agent $i$'s value, or the large items contribute at most a constant fraction of agent $i$'s value. Hence, in both cases we can simply discard the large items in the analysis and lose only a constant factor. 

We remark that in the algorithm, we apply randomized rounding to {\em all} items in $\cG \setminus \cH$, without distinguishing large and small items. This does not hurt and the algorithm is more natural this way; in any case the large items may provide some additional value. Also, we do not find the particular matching $\rho$ described here; we simply find the most profitable matching at the end. This provides a solution at least as good as the one we analyze in our proof.

\medskip

In the following, we describe each phase in detail.

\subsection{Phase I: Initial Matching}

First, we solve the Nash Social Welfare problem under the restriction that we only allocate at most one item to each agent. To achieve this, consider the complete bipartite graph between $\cA$ and $\cG$ and assign an edge weight $\omega_{ij} = \log v_i(j)$ to every edge $(i, j)\in \cA\times\cG$. We can find an optimal assignment $\tau: \cA\to\cG$ by computing the maximum-weight matching in this bipartite graph; i.e., $\tau(i)$ is the item matched to agent $i$. We define $\cH = \tau(\cA)$ to be the set of matched items.
We note that each item in the matching has positive value $v_{i}(\tau(i))> 0$ for the respective agent, otherwise there is no matching of positive value, which means that $OPT=0$.

\subsection{Phase II: Mixed Multilinear Relaxation} 

In this section, we describe our new ``Mixed Multilinear" relaxation for the Nash Social Welfare Problem, and a restricted ``Matching+Multilinear" version of it, which we show to be within a constant factor of each other. Although these relaxations are new, they are naturally analogous to the relaxations in  \cite{garg2020approximating}.

\paragraph{Mixed Multilinear Relaxation.}
For each valuation $v_i:2^\cG \to \RR_+$, we define its multilinear extension $V_i: [0,1]^{\cG} \to \RR_+$ as
$$V_i(\by_i) = \sum_{S \subseteq \cG} v_i(S) \prod_{j \in S} y_{ij} \prod_{j' \in \cG' \setminus S} (1-y_{ij'}).$$
We propose the following relaxation of the Nash Social Welfare problem.

\begin{align*}
\tag{Mixed-Multilinear}
\label{eqns:MM}
\max \quad & \prod_{i \in \cA} V_i(\bx_i) \\
s.t. \quad & \sum_{i\in\cA} x_{ij} \le 1 && \forall j\in \cG \\
& x_{ij} \ge 0 &&  \forall i \in \cA, j \in \cG \setminus \cH \\
& x_{ij} \in \{0,1\} && \forall i \in \cA, j \in \cH
\end{align*}

Note that although $\cH$ was chosen by matching one item to each agent, this might not be the case in the optimal solution. Indeed in (\ref{eqns:MM}), we allow $\cH$ to be allocated arbitrarily; but the assignment cannot be fractional. (If we allowed all items to be assigned fractionally, the relaxation would have an infinite integrality gap, for well-known reasons.) This relaxation is difficult to deal with, because it's hard to find a good assignment of $\cH$. Instead, just like in \cite{garg2020approximating}, we consider a restricted version of this relaxation, where $\cH$ is required to be allocated by a matching.

\begin{align*}
\tag{Matching+Multilinear}
\label{eqns:mixed+matching}
\max \quad & \prod_{i\in\cA} V_i(\by_i + \b1_{\sigma(i)}) \\
s.t. \quad & \sum_{i\in\cA} y_{ij}\le 1 && \forall j\in \cG\backslash\cH \\
& y_{ij} \ge 0 && \forall i \in \cA, j \in \cG \setminus \cH \\
& y_{ij} = 0 && \forall i \in \cA, j \in \cH \\
& \sigma: \cA\to\cH \mbox{ is a matching.} && 
\end{align*}

Denote by $\OPT$ the optimum value of (\ref{eqns:MM}), and by $\OPT_{\cH}$ the optimal value of the above program (\ref{eqns:mixed+matching}).  
Similar to Theorem 3.2 in \cite{garg2020approximating}, we have:

\begin{lemma}
\label{lem:1/2}
$$\OPT_{\cH}\ge \frac{1}{3^{1/3}} \OPT.$$
\end{lemma}

\begin{proof}
Consider an optimum solution $\bx^*$ of (\ref{eqns:MM}), that is $\bx^*_i = \by^*_i + \b1_{H^*_i}$ where $\by^* \in [0,1]^{\cA \times \cG}$ is a fractional assignment of the items in $\cG'$ and $(H^*_1,\ldots,H^*_n)$ is a partition of $\cH$. We construct a feasible solution $(\by^*, \sigma)$ for (\ref{eqns:mixed+matching}), where $\sigma:\cA \rightarrow \cH$ is a matching such that for $H^*_i \neq \emptyset$, $\sigma(i)$ is the most valuable item in $H^*_i$, and the remaining items in $\cH$ are matched arbitrarily to agents such that $H^*_i = \emptyset$.

Let $k_i = |H^*_i|$ be the number of $\cH$-items allocated to agent $i$ in the optimal solution. If $k_i > 0$, $\sigma(i)$ is the most valuable of them, and by submodularity $v_i(H^*) \leq k_i v_i(\sigma(i))$. This also implies $V_i(\by^*_i + \b1_{H^*_i}) \leq \max \{k_i, 1\} V_i(\by^*_i + \b1_{\sigma(i)})$. 
Hence, we can write
\begin{eqnarray*}
    \OPT &=& \left( \prod_{i\in\cA} V_{i}(\by_i+ \b1_{H^*_i}) \right)^{1/n}  \\
    & \leq & \left( \prod_{i \in \cA} \max \{k_i, 1\} \ V_i(\by^*_i + \b1_{\sigma(i)}) \right)^{1/n} \\
    & \leq & \left( \prod_{i \in \cA} 3^{k_i/3} \ V_i(\by^*_i + \b1_{\sigma(i)}) \right)^{1/n} \\
    & = & 3^{1/3} \left( \prod_{i \in \cA} V_i(\by^*_i + \b1_{\sigma(i)}) \right)^{1/n} \\
    & \leq & 3^{1/3} \ OPT_\cH 
\end{eqnarray*}
where we used the AMGM inequality, the fact that $\max \{k,1\} \leq 3^{k/3}$ for every integer $k \geq 0$, and $\sum_{i=1}^{n} k_i = n$.
\end{proof}

We remark that the factor of $3^{1/3}$ is tight due to the following instance: $|\cH| = |\cG \setminus \cH| = n$, $n/3$ agents have the valuation $v(S) = |S \cap \cH|$, and the remaining $2n/3$ agents have the valuation $v'(S) = \min \{|S|, 1\}$. The optimal Nash Social Welfare is $3^{1/3}$, since $n/3$ agents can get value $3$ from $3$ items of $\cH$ each, and the remaining agents get value $1$ from items in $\cG \setminus \cH$. If $\cH$ is allocated as a matching, we get Nash Social Welfare $1$, since each agent receives value $1$.

\subsection{Phase III: The Iterated Continuous Greedy Algorithm}
\label{sec:iter-greedy}

In this section, we describe the details of Phase III where we aim to find a fractional solution of our (LogMultilinear) relaxation of Nash Social Welfare. We do this for a subset of items $\cG' = \cG \setminus \cH$, and a subset of agents $\cA'$ who have positive value for these items.

\begin{align*}
\tag{LogMultilinear}
\max \quad & \frac1n\sum_{i\in\cA'} \log V_i(\by_i) \\
s.t. \quad & \sum_{i\in\cA'} y_{ij}\le 1 && \forall j\in \cG' \\
&  y_{ij} = 0 && \forall i \in \cA', j \in \cH \\
& \by\ge 0 &&  
\end{align*}

We recall that 
$V_i(\by_i) = \sum_{S \subseteq \cG} v_i(S) \prod_{j \in S} y_{ij} \prod_{j' \in \cG \setminus S} (1-y_{ij'})$ is the multilinear extension of $v_i$. In this section we assume that the vector $\by_i$ always has $0$ in coordinates indexed by $j \in \cH$, so effectively we are working with vectors in $[0,1]^{\cG'}$.

We design a variant of the continuous greedy algorithm which approximates the optimal solution within an additive error of $1$.

\paragraph{The Iterated Continuous Greedy Algorithm}
\begin{enumerate}
    \item Start with a feasible solution $\by^{(0)}$, $y^{(0)}_{ij} = \frac{1}{n}$ for each $i \in \cA'$ and $j \in \cG'$.
    \item For $r = 1,2,\ldots$, given a feasible solution $\by^{(r-1)}$, initiate $\by(\frac12) = \frac12 \by^{(r-1)}$ and run the following continuous greedy algorithm:
    \begin{itemize}
        \item Let $\bz(t)$ be a feasible solution (satisfying $\bz \geq 0$ and $\sum_i z_{ij} \leq 1$ for each $j$) which maximizes the linear objective function
        $$ \sum_{i \in \cA'} \frac{\bz_i \cdot \nabla V_i(\by_i(t))}{V_i(\by_i(t))}.$$
        \item Evolve the solution $\by(t)$ according to the equation
        $$ \frac{d}{dt} \by(t) = \bz(t),$$
        for $t \in [\frac12,1]$.
    \end{itemize}
    
    \item Set $\by^{(r)} = \by(1)$, the solution obtained in this iteration. 
    
    \item If $\frac{1}{n} \sum_{i \in \cA'} \log V_i(\by^{(r)}) \geq \frac{1}{n} \sum_{i \in \cA'} \log V_i(\by^{(r-1)}) + \frac18$, let $r \leftarrow r+1$ and repeat. 
    
    \item Otherwise, return $\by^{(r)}$.
\end{enumerate}

\begin{theorem}
\label{thm:optimization}
Let $\by^*$ denote any feasible solution of the optimization program (LogMultilinear). 
Assuming that $v_i(\cG') > 0$ and $v_i$ is monotone submodular for each $i \in \cA'$, the Iterated Continuous Greedy algorithm terminates in $O(\log n)$ iterations and returns a feasible solution $\by$ for (LogMultilinear) such that
$$ \frac{1}{n} \sum_{i \in \cA'} \frac{V_i(\by^*_i)}{V_i(\by_i)} \leq e.$$
\end{theorem}

We note that by concavity of the logarithm, the conclusion also implies $\frac{1}{n} \sum_{i \in \cA'} \log  \frac{V_i(\by^*_i)}{V_i(\by_i)} \leq 1$, i.e.~our solution approximates the optimum of (LogMultilinear) within an additive error of $1$. The statement in the lemma is stronger and more convenient, though, which we will use later in several places. 

\begin{proof}
As a starting point, we have $y^{(0)}_{ij} = \frac{1}{n}$. By concavity of $V_i$ in positive directions, we have the simple bound $V_i(\by^{(0)}_i) \geq \frac{1}{n} V_i(\b1)$. Hence, $\frac{1}{n} \sum_{i \in \cA'} \log V_i(\by^{(0)}_i) \geq \frac{1}{n} \sum_{i \in \cA'} \log V_i(\b1) - \log n$.
Now we apply the continuous greedy algorithm as above, and we iterate as long as after each iteration we have $\frac{1}{n} \sum_{i \in \cA'} \log V_i(\by^{(r)}) \geq \frac{1}{n} \sum_{i \in \cA'} \log V_i(\by^{(r-1)}) + \frac18$. Since for any feasible solution, $\sum_{i \in \cA'} \log V_i(\by^{(r)}) \leq \sum_{i \in \cA'} \log V_i(\b1)$, this means that we cannot iterate more than $O(\log n)$ times. It remains to prove that the solution satisfies the claimed inequality.

To prove this, assume at any time $t$ that
\begin{equation}
\tag{*}
 \sum_{i \in \cA'} \frac{V_i(\by^*_i)}{V_i(\by_i)} > e n. 
\end{equation}
A possible direction for the continuous greedy algorithm to pursue is always $\bz = \by^*$. For this direction, 
we obtain
$$ \sum_{i \in \cA'} \frac{\by^*_i \cdot \nabla V_i(\by_i(t))}{V_i(\by_i(t))} 
\geq \sum_{i \in \cA'} \frac{V_i(\by^*_i) - V_i(\by_i(t))}{V_i(\by_i(t))}
= \sum_{i \in \cA'} \left( \frac{V_i(\by^*_i)}{V_i(\by_i(t))} - 1 \right) > (e-1) n $$
using the monotonicity and concavity of $V_i$ in nonnegative directions in the first inequality,
and our assumption (*) in the second inequality.
Since the continuous greedy algorithm chooses a direction $\bz_i(t)$ by optimizing the expression
$ \sum_{i \in \cA'} \frac{\bz_i \cdot \nabla V_i(\by_i(t))}{V_i(\by_i(t))}$,
we obtain the same bound for the greedy direction $\bz_i(t)$, and finally by the chain rule we have
$$ \frac{d}{dt} \sum_{i \in \cA'} \log V_i(\by_i(t)) = \sum_{i \in \cA'} \frac{1}{V_i(\by_i(t))} 
\nabla V_i(\by_i(t)) \cdot \frac{d\by_i}{dt} = \sum_{i \in \cA'} \frac{\bz_i(t) \cdot \nabla V_i(\by_i(t))}{V_i(\by_i(t))} > (e-1)n.$$
Hence the rate of increase in $\sum_{i \in \cA'} \log V_i(\by_i(t))$ is at least $(e-1) n$ as long as (*) is satisfied.

Each iteration starts by scaling the previous solution by a factor of $\frac12$ and then
running continuous greedy for $t$ between $\frac12$ and $1$. Again by concavity, we have 
$V_i(\frac12 \by^{(r-1)}) \geq \frac12 V_i(\by^{(r-1)})$. 
By integration over the course of the continuous greedy process, we obtain
$$ \sum_{i \in \cA'} \log V_i(\by^{(r)}_i) = \sum_{i \in \cA'} \log V_i\left(\frac12 \by^{(r-1)}_i\right)
+ \int_{1/2}^{1} \frac{d}{dt} \sum_{i \in \cA'} \log V_i(\by_i(t)) dt $$
$$ \geq \sum_{i \in \cA'} \log \left( \frac12 V_i(\by^{(r-1)}_i) \right) + \int_{1/2}^1 (e-1) n \ dt $$
$$ = \sum_{i \in \cA'} \log V_i(\by^{(r-1)}_i) + \left(\frac{e-1}{2} - \log 2 \right) n.$$
We note that all logarithms here are natural and $\frac{e-1}{2} - \log 2 > \frac18$. 
Hence we gain at least $\frac18 n$ in each iteration as long as (*) is satisfied, and we terminate otherwise.
\end{proof}

\paragraph{Discretization.}
As in the original continuous greedy algorithm \cite{CCPV11}, we need to discretize the continuous process to obtain an actual polynomial-time algorithm. This can be done using standard methods. 

First, for any given $\by_i(t)$, we can estimate by random sampling 
$$\partdiff{V_i}{y_j}\Big|_{\by_i(t)} = \E[v_i(R_i(t) + j) - v_i(R_i(t) - j)]$$
where $R_i(t)$ is a random set containing each item $j$ independently with probability $y_{ij}(t)$. Since $v_i(R_i(t)+j) - v_i(R_i(t)-j) \in [0, v_i(\{j\})]$, using $poly(m,n)$ samples we can obtain estimates $\omega_{ij}$ of $\partdiff{V_i}{y_j}$ within an error of $\frac{v_i(\{j\})}{poly(m,n)}$ with high probability.

Then we find a direction $\bz(t)$ by solving the linear programming problem 
$$\max \left\{ \sum_{i \in \cA'} \frac{1}{V_i(\by_i(t))} \sum_{j \in \cG'} \omega_{ij} z_{ij}: z_{ij} \geq 0, \sum_i z_{ij} \leq 1 \ \forall j \right\} $$
(using $\omega_{ij}$ in place of $\partdiff{V_i}{y_j}$). If the estimates $\omega_{ij}$ are correct up to an error of $\frac{v_i(\{j\})}{poly(m,n)}$, the optimum is correct up to a relative error of $\frac{1}{poly(m,n)}$. Note that $V_i(\by_i(t)) \geq \frac{1}{poly(m,n)} \sum_{i \in \cG'} v_i(\{j\})$ since this is true for the initial solution $\by^{(0)}$ and the value can only decrease $O(\log n)$ times by a factor of $2$; apart from that it increases.

Then we make a step of size $\delta = \frac{1}{poly(m,n)}$, where we set $\by(t+\delta) = \by(t) + \delta \cdot \bz(t)$. 
The guarantee we claim here is that
$$ \sum_{i \in \cA'} \log V_i(\by_i(t+\delta)) \geq \sum_{i \in \cA'} \log V_i(\by_i(t)) + \delta \, \left(1 - \frac{1}{poly(m,n)} \right) \sum_{i \in \cA'} \frac{\by^*_i \cdot \nabla V_i(\by_i(t))}{V_i(\by_i(t))}.  $$
This is true because we find the optimum of the linear programming problem within a $\frac{1}{poly(m,n)}$ relative error, and also the values $V_i(\by_i)$ and the partial derivatives $\partdiff{V_i}{y_j}$ can change only by a factor of $1 \pm \frac{1}{poly(m,n)}$ between $\by_i(t)$ and $\by_i(t+\delta)$, as long as $t \leq 0.99$ (since $V_i(\by_i)$ and $\partdiff{V_i}{y_j}$ are nonnegative and linear in each coordinate separately). Hence, we can mimic the continuous analysis for $t \in [0.5,0.99]$ within a $\frac{1}{poly(m,n)}$ relative error at every step, and we lose a factor of $49/50$ by ignoring the improvement between $[0.99, 1]$. These errors are easily absorbed for example in the gap between $\frac{e-1}{2} - \log 2$ and $\frac18$ which we ignore above. So the theorem still holds for the discretized algorithm, with high probability.

\subsection{Phase IV: Randomized Rounding}
\label{sec:rand-rounding}

In this section, our goal is to round the fractional solution $\by$ from Section~\ref{sec:iter-greedy}.
In the actual algorithm, we use the following simple randomized rounding procedure.

\paragraph{RandomizedRounding($\by$)}
\begin{enumerate}
    \item For each item $j \in \cG'$ independently, select $Z_j \in \{0,1,\ldots,n\}$ where $Z_j = i$ with probability $y_{ij}$, or $Z_j = 0$ with probability $1 - \sum_{i \in \cA'} y_{ij}$.
    
    \item Define $R_i = \{j \in \cG \setminus \cH: Z_j = i \}$.

    \item Return $(R_1,\ldots,R_n)$.
\end{enumerate}

However, in the analysis we will proceed more carefully, separating the contributions of ``large'' and ``small'' items. We first define what we mean by ``large'' and ``small''. 
For any $\by \in [0,1]^{\cG}$ and $S\subseteq \cG$, define vector $\by^{(S)}$ to be the vector obtained by setting all the coordinates not in $S$ to $0$.
\begin{align*}
y^{(S)}_i = \left\{ \begin{array}{ll}
y_{i} & i\in S \\
0  & i\notin S.\\
\end{array}\right.
\end{align*}
For each agent $i \in \cA'$, we define the set $L_i$ of ``large items'' as follows, for a given constant $c>0$.
Let us assume in the following that $\sum_{j \in \cG'} y_{ij} \geq c$ for every agent $i$.
This is without loss of generality, since we can always extend the instance with dummy items of value $0$, which can be allocated fractionally to any agent and it doesn't change the outcome of our algorithm in any way.

\paragraph{FindLargeSet($i$, $\by$)}
\begin{enumerate}
    \item Start with an empty set at time $0$, $L^{(0)}_i = \emptyset$. 
    
    \item At time $t \geq 1$, add the item with the largest marginal value to $L^{(t-1)}_i$. More specifically, let $L^{(t)}_i = L^{(t-1)}_i \cup \{j_t\}$ where 
    \begin{align*}
    j_t &= \argmax_{j \in \cG' \backslash L^{(t-1)}_i} \left(V_i(\by_i^{(L^{(t-1)}_i)}+\vect{1}_{j}) - V_i(\by_i^{(L^{(t-1)}_i)}\right) 
    \end{align*}
    
    \item As long as $\sum_{t'=1}^t y_{i j_{t'}} < c$ and $\cG' \setminus L_i^{(t)} \neq \emptyset$,
    let $t \leftarrow t+1$ and repeat step 2.
    
    \item Return $L_i := L^{(t)}_i$.
\end{enumerate}




\noindent 
We have two simple corollaries for the set $L_i$.
\begin{itemize}
\item For any agent $i\in\cA'$, $$c \leq \sum_{j\in L_i}y_{ij} < c+1,$$
\item For any $i\in\cA'$, $j \in \cG' \backslash L_i$, 
$$ V_i(\by_i^{(L_i)}+\vect{1}_{j}) - V_i(\by_i^{(L_i)}) \le \frac{1}{c} V_i(\by_i^{(L_i)}).$$
\end{itemize}

The first property follows from the stopping rule (including our assumption that each agent gets $\sum_{j \in \cG'} y_{ij} \geq c$ in the fractional solution). As for the second one, if the marginal value is 0 for any $j \in \cG' \backslash L^{(t-1)}_i$, it is trivially true. Otherwise, consider any item $j$ in $L_i$ that we did not include in the procedure; (by submodularity) in every step we included an item $j_t$ of marginal value $V_i(\by_i^{(L^{(t-1)}_i)}+\vect{1}_{j_t}) - V_i(\by_i^{(L^{(t-1)}_i)}) \geq V_i(\by_i^{(L_i)}+\vect{1}_{j}) - V_i(\by_i^{(L_i)})$ and by multilinearity the total contribution of the included items is 
$$V_i(\by_i^{(L_i)}) =  \sum_{t=1}^{|L_i|} y_{i j_t}  (V_i(\by_i^{(L^{(t-1)}_i)}+\vect{1}_{j_t}) - V_i(\by_i^{(L^{(t-1)}_i)}))\geq c (V_i(\by_i^{(L_i)}+\vect{1}_{j}) - V_i(\by_i^{(L_i)})).$$

Now we can describe our modified rounding procedure. We note that this procedure is used only in the analysis.

\paragraph{RestrictedRandomizedRounding($\by$)}
\begin{enumerate}
    \item Compute the set $L_i$ (specified above) for each agent $i\in\cA'$.
    
    \item For each item $j \in \cG'$, assign $j$ to a random player according to $y_{ij}$: \\
    Let $Z_j = i$ with probability $y_{ij}$, or $Z_j = 0$ with probability $1 - \sum_{i \in \cA'} y_{ij}$. \\
    For each $i \in \cA'$, let $S_i = \{j\in\cG'\backslash L_i: Z_j = i\}$ and $\by^{(s)}_i =  \by_i^{(L_i)}+\vect{1}_{S_i}$.
    
    \item Return $\by^{(s)}$.
\end{enumerate}

Note that only ``small items'' are included in the sets $S_1,\ldots,S_n$, and large items are still assigned fractionally in $\by^{(s)}$. Thus the solution $\by^{(s)}$ can be viewed as ``sparsified'' rather than rounded. We note that the notion of sparsity here is in terms of the summation of fractional variables ($\sum_{i \in \cA'} \sum_{j \in L_i} y^{(s)}_{ij} < (c+1) n$) rather than the size of the support of $\by^{(s)}$.

The notion of large/small is agent-specific, so $\by^{(s)}$ might not even be a feasible solution; an item could be allocated fully as a small item and still fractionally as a large item for other agents. We will show at the end that large items can be in fact discarded. However, for now we analyze the value of $\by^{(s)}$.


\begin{lemma}
\label{lem:small-items} 
Suppose that 
$$ \frac{1}{n} \sum_{i \in \cA'} \frac{V_i(\by^*_i)}{V_i(\by_i)} \leq \alpha. $$
Then with probability $\Omega(\epsilon)$, the solution $\by^{(s)}$ obtained by {\bf RestrictedRandomizedRounding}($\by$) with parameter $c>0$ satisfies
$$ \frac{1}{n} \sum_{i \in \cA'} \frac{V_i(\by^*_i)}{V_i(\by^{(s)}_i)} \leq (1+\epsilon) (2 + 4/c)  \alpha. $$
\end{lemma}

\begin{proof}
Using the notation from {\bf RestrictedRandomizedRounding}($\by$),
for every $i\in\cA'$, we define a monotone submodular function $u_i:2^{\cG'\backslash L_i} \to \RR$, where $u_i(S) = V_i(\by_i^{(L_i)}+\vect{1}_{S})$. Recall that $\by^{(s)}_i = \by_i^{(L_i)} + \b1_{S_i}$; that is, $V_i(\by_i^{(s)}) = u_i(S_i)$. The sets $S_1,\ldots,S_n$ are determined by the random variables $(Z_j: j \in \cG')$. 
Our goal is to upper-bound 
$$V(\cZ) = V(Z_j: j \in \cG')
 =  \frac1n \sum_{i\in\cA'}\frac{V_i(\by_i^*)}{V_i(\by^{(s)}_i)}
  = \frac1n \sum_{i\in\cA'} \frac{V_i(\by_i^*)}{u_i(S_i)}.$$
By the definition of $L_i$ and by submodularity, we know that for any $i\in\cA'$, $j\in\cG'\backslash L_i$ and $S\subseteq\cG'\backslash L_i$, 
$$ 0 \leq u_i(S\cup\{j\}) - u_i(S)\le \frac{V_i(\by_i^{(L_i)})}{c} = \frac{u_i(\emptyset)}{c}.$$ 
Since $u_i(S_i)$ is a function of the independent random variables $(Z_j: j \in \cG')$, by the Efron-Stein inequality, we have 
\begin{align*}
    \Var[u_i(S_i)] 
    &\le  \E\left[\sum_{j \in \cG'} \left(u_i(S_i) - \min_{Z_j} u_i(S_i)\right)^2\right] \\
    & =  \E\left[\sum_{j\in S_i}\left(u_i(S_i) - u_i(S_i\backslash\{j\})\right)^2\right] \\
    &\le  \frac{u_i(\emptyset)}{c}\cdot\E\left[\sum_{j\in S_i}\left(u_i(S_i) - u_i(S_i\backslash\{j\})\right)\right] \\
    &\le  \frac{u_i(\emptyset)}{c}\cdot\E[u_i(S_i)] = \frac{u_i(\emptyset)}{c} \cdot V_i(\by_i)
\end{align*} 
where we used the submodularity of $u_i$ in the last inequality.
By Chebyshev's inequality, we have
$$\Pr\left[u_i(S_i)\le \frac{V_i(\by_i)}{2}\right] \le \frac{\Var[u_i(S_i)]}{(V_i(\by_i)/2)^2}\le \frac{4u_i(\emptyset)}{c \, V_i(\by_i)}.$$
Therefore,
$$\E\left[\frac{V_i(\by_i)}{u_i(S_i)}\right] \le \frac{V_i(\by_i)}{V_i(\by_i)/2} + \frac{V_i(\by_i)}{u_i(\emptyset)}\cdot\Pr\left[u_i(S_i)\le \frac{V_i(\by_i)}{2}\right]\le 2 + \frac{4}{c}. $$ 
Combining this with $\frac1n\sum_{i\in\cA'}\frac{V_i(\by_i^*)}{V_i(\by_i)} \leq \alpha$, we can write
$$ \E[V(\cZ)] = \frac1n\sum_{i\in\cA'}\frac{V_i(\by_i^*)}{V_i(\by_i)} \E\left[ \frac{V_i(\by_i)}{u_i(S_i)} \right] \leq  (2+4/c) \alpha.$$
By Markov's inequality, we conclude that with probability $\Omega(\epsilon)$, $V(\cZ)\le (1+\epsilon) (2+4/c) \alpha$.
\end{proof}

\subsection{Phase V: Matching recombination}
\label{sec:rematching}

Now we have a fractional solution $\by^{(s)}$ with good properties; however, we ignored the fact that $\cH$ should be also allocated. Our goal in this section is to prove that there exists a matching which works well with our fractional solution $\by^{(s)}$, and at the same time it has additional properties which allow us round the large items (or in fact discard them!) and still obtain a good value of Nash Social Welfare. 

We proceed very much as in \cite{garg2020approximating}. First, we prove that there exists a matching $\sigma$ which obtains a good value together with $\by^{(s)}$. 

\paragraph{Matching extension.}
Here we show that there exists a matching $\sigma:\cA \rightarrow \cH$ which complements well the fractional solution $\by^{(s)}$.

\begin{lemma}
\label{lem:rematching}
Let $\bx^*$ be the optimal solution of (MixedMultilinear), i.e.~$\bx^*_i = \by^*_i + \b1_{H^*_i}$ where $\by^* \in [0,1]^{\cA \times \cG}$ is a feasible solution of (LogMultilinear) and $(H^*_1,\ldots,H^*_n)$ is a partition of $\cH$.
Let $\by' \in [0,1]^{\cA \times \cG}$ be an arbitrary fractional solution, satisfying
$$ \frac{1}{n} \sum_{i \in \cA'} \frac{V_i(\by^*_i)}{V_i(\by'_i)} \leq \beta $$
and $\by'_i = 0$ for $i \notin \cA'$.
Then there is a matching $\pi:\cA \rightarrow \cH$ such that 
$$ NSW(\by',\pi) = \left(\prod_{i \in \cA} V_i(\by'_i + \b1_{\pi(i)}) \right)^{1/n}
 \geq \frac{1}{\beta+1} \left( \prod_{i \in \cA} V_i(\bx^*_i) \right)^{1/n} = \frac{1}{\beta+1} \ OPT.$$
\end{lemma}

\begin{proof}
Suppose that agent $i$ receives $k_i = |H^*_i|$ items from $\cH$ in the optimal solution, and let $\pi(i) \in H^*_i$ be the most valuable item in $H^*_i$ (as a singleton).  We extend this to a matching $\pi:\cA \rightarrow \cH$, by allocating any remaining items arbitrarily to agents such that $H^*_i = \emptyset$. 
By the AMGM inequality, we can write
$$ \frac{OPT}{NSW(\by,\pi)} = \left( \prod_{i \in \cA} \frac{V_i(\bx^*_i)}{V_i(\by_i + \b1_{\pi(i)})} \right)^{1/n}
\leq \frac{1}{n} \sum_{i \in \cA} \frac{V_i(\bx^*_i)}{V_i(\by_i + \b1_{\pi(i)})} $$
By submodularity, we have $V_i(\bx^*_i) = V_i(\by^*_i + \b1_{H^*_i}) \leq V_i(\by^*_i) + k_i V_i(\b1_{\pi(i)})$.
Thus, we obtain
\begin{align*}
    \frac{OPT}{NSW(\by,\pi)}
&\leq \frac{1}{n} \sum_{i \in \cA} \frac{V_i(\by^*_i) + k_i V_i(\b1_{\pi(i)})}{V_i(\by_i + \b1_{\pi(i)})}  \\
 &\leq \frac{1}{n} \left( \sum_{i \in \cA'}\frac{V_i(\by^*_i)}{V_i(\by_i + \b1_{\pi(i)})} + \sum_{i \in \cA} k_i \right) \\
 &\leq \frac{1}{n} \left( \sum_{i \in \cA'}\frac{V_i(\by^*_i)}{V_i(\by_i)} + \sum_{i \in \cA} k_i \right)
\end{align*} 
using monotonicity of $V_i$ in the denominator. Note that $V_i(\by^*_i) = 0$ and $v_i(\pi(i)) > 0$ for every agent $i \notin \cA'$, because these agents do not derive any value from $\cG' = \cG \setminus \cH$ and hence $v_i(H^*_i) > 0$ for these agents; that's why we can switch to $\cA'$ in the first summation. Finally, using the assumption $\frac{1}{n} \sum_{i \in \cA'} \frac{V_i(\by^*_i)}{V_i(\by_i)} \leq \beta$
and the fact that $\frac{1}{n} \sum_{i \in \cA} k_i = \frac{1}{n} \sum_{i \in \cA} |H^*_i| = 1$, we obtain
$$ \frac{OPT}{NSW(\by,\pi)} \leq \beta + 1.$$
 \end{proof}

\begin{corollary}
\label{cor:rematching}
The fractional solution $\by^{(s)}$ = {\bf RestrictedRandomizedRounding}($\by$) satisfies with constant probability
$$ \max_\pi \NSW(\by^{(s)},\pi) \geq \frac{1}{7+12/c} OPT.$$
\end{corollary}

\begin{proof}
Since $\by$ from the Iterated Continuous Greedy algorithm satisfies $\frac{1}{n} \sum_{i \in \cA'} \frac{V_i(\by^*_i)}{V_i(\by_i)} \leq e$, we apply Lemma~\ref{lem:small-items} with $\alpha = e$.
For $\epsilon = 3/e - 1$, we get $ \frac{1}{n} \sum_{i \in \cA'} \frac{V_i(\by^*_i)}{V_i(\by^{(s)}_i)} \leq 3 (2 + \frac{4}{c})$ with constant probability. Then, we apply Lemma~\ref{lem:rematching} with $\by' = \by^{(s)}$ and $\beta = 3 (2 + \frac{4}{c})$.
We conclude that there is a matching $\pi$ such that $\NSW(\by,\pi) \geq \frac{1}{7 + 12/c} OPT$.
\end{proof}

\paragraph{Matching recombination.}
Now that we know a good matching exists, we want to show that there exists another matching $\rho$ with some additional desirable properties. The matching $\rho$ should be such that each agent $a$ either gets significant value from the matching item $\rho(a)$ alone, or there is no item of very large value contributing to agent $a$ in the fractional solution. The solution is a procedure we borrow almost verbatim from \cite{garg2020approximating}:
a careful combination of the initial matching $\tau$ and a matching $\pi$ optimal with respect to our fractional solution $\by$. Our goal is to prove the following lemma, analogous to Lemma 6.1 in \cite{garg2020approximating}. Since our setup here is somewhat different, we repeat the whole argument in a self-contained manner. Also, we remark that while this is an actual algorithmic step in \cite{garg2020approximating}, we only need this procedure in the analysis.

\begin{lemma}
\label{lem:matching}
Let $d \geq 2$.
Let $\tau:\cA \rightarrow \cG$ be the matching maximizing $\prod_{a \in \cA} v_i(\tau(a))$, $\cH = \tau(\cA)$ the items allocated in this matching. Let $\by \in [0,1]^{\cA' \times \cG'}$ and let $\pi: \cA \rightarrow \cH$ be any matching. Then there is a matching $\rho: \cA \rightarrow \cH$ such that
$$ \NSW(\by,\rho) \geq \frac{1}{d+2} \NSW(\by, \pi) $$
and for every agent $a \in \cA$,
\begin{enumerate}
    \item [(i)] either $v_a(\rho(a)) \geq \frac{1}{d} V_a(\by_a)$ 
    (in which case the $\rho$-matching item itself recovers a constant fraction of agent $a$'s value)
    
    \item [(ii)] or for every item $j \in \cG'$, $v_a(j) < \frac{1}{d} V_a(\by_a)$  
    (in which case there are no items with large contributions to $V_a(\by_a)$).
\end{enumerate}
\end{lemma}

\begin{proof}
Let $\tau$ be the initial optimal matching and $\cH = \tau(\cA)$.
Let $\by \in [0,1]^{\cA' \times \cG'}$ and let $\pi: \cA \rightarrow \cH$ be any matching. (We will use the optimal matching with respect to $\by$ but that is not relevant now.)

We will construct a new matching $\rho$ which combines $\tau$ and $\pi$ in a certain way. 
First, whenever $\tau(a) = \pi(a)$, we set $\rho(a) = \tau(a) = \pi(a)$. 
Next, we consider the two matchings as sets of edges $(a,\pi(a))$ and $(a,\tau(a))$ and consider their symmetric difference, $\pi \Delta \tau$. The symmetric difference consists of alternating paths and cycles covering the agents such that $\pi(a) \neq \tau(a)$.

Let $\cB = \{ a \in \cA: v_a(\pi(a)) < \frac{1}{d-1} V_a(\by_a) \}$. We define a modified matching $\pi'$ where $\pi'(a) = \pi(a)$ for $a \notin \cB$ and $\pi'(a) = \emptyset$ for $a \in \cB$, meaning that agents $a \in \cB$ don't get any items in $\pi'$. If $\pi(a)$ contributes less than $\frac{1}{d-1} V_a(\by_a)$, we have $V_a(\by_a + \b1_{\pi(a)}) \leq V_a(\by_a) + v_a(\pi(a)) \leq \frac{d}{d-1} V_a(\by_a) \leq \frac{d}{d-1} V_a(\by_a + \b1_{\pi'(a)})$, and so
\begin{equation}
    \label{eq:pi'}
\NSW(\by, \pi') = \left(\prod_{a \in \cA} V_a(\by_a + \b1_{\pi'(a)})\right)^{1/n} \geq \left(\prod_{a \in \cA} \frac{d-1}{d} V_a(\by_a + \b1_{\pi(a)})\right)^{1/n} = \frac{d-1}{d}\NSW(\by, \pi).
\end{equation}

Consider an alternating path/cycle $C$ in $\pi \Delta \tau$ and its set of agents $\cA(C)$. 
We distinguish two cases.

\begin{enumerate}
\item $\cB \cap \cA(C) = \emptyset$ ($\pi$ provides good value for all agents in $\cA(C)$). In this case we set $\rho(a) = \pi(a)$ for all $a \in \cA(C)$.

\item $\cB \cap \cA(C) \neq \emptyset$ (some agents in $\cA(C)$ don't get good value from $\pi$). We remove from $C$ every edge $(a,\pi(a))$ such that $a \in \cB$ (which means that $\pi'(a) = \emptyset$); this breaks $C$ into alternating paths. Let us consider one such alternating path, denoting the agents on it $a_1,a_2,\ldots,a_k$ and the items $i_1,i_2,\ldots,i_k$. If $k=1$, the path consists of just one edge $(a_1,i_1)$. If $k>1$, the path consists of edges $(a_1,i_1),  (i_1,a_2), (a_2, i_2), \ldots, (a_k,i_k)$, where $i_j = \tau(a_j)$ for $j \leq k$ and $i_j = \pi(a_{j+1})$ for $j<k$.  We also have $a_1 \in \cB$ (this is an agent who does not get any item in $\pi'$) and $a_2,\ldots,a_k \notin \cB$. 

We use the following criterion to decide whether we should use the $\pi$-edges or the $\tau$-edges from this alternating path: Let
\begin{align*}
    \varphi(a_1,\ldots,a_k) 
&= \frac{V_{a_1}(\by_{a_1})}{V_{a_1}(\by_{a_1} + \b1_{\tau(a_1)})} \prod_{j=2}^{k} \frac{V_{a_j}(\by_{a_j} + \b1_{\pi(a_j)})}{V_{a_j}(\by_{a_j} + \b1_{\tau(a_j)})} \\
&= \frac{V_{a_1}(\by_{a_1})}{V_{a_1}(\by_{a_1} + \b1_{i_1})} \prod_{j=2}^{k} \frac{V_{a_j}(\by_{a_j} + \b1_{i_{j-1}})}{V_{a_j}(\by_{a_j} + \b1_{i_j})}.
\end{align*}  
This is the factor incurred in the objective function if we switch from $\pi'$ to $\tau$ on this alternating path. We call this alternating path {\em $\tau$-favorable}\footnote{``reversible'' in \cite{garg2020approximating}}, if $\varphi(a_1,\ldots,a_k) \leq d^k$, and we define $\rho(a_j) = \tau(a_j) = i_j$ for $1 \leq j \leq k$. Otherwise, we call it $\pi$-favorable and we define $\rho(a_1) = \emptyset$, $\rho(a_j) = \pi(a_j) = i_{j-1}$ for $2 \leq j \leq k$.

If we view the process as starting from the matching $\pi'$ and then applying a swap for each $\tau$-favorable path, we obtain a solution $(\by, \rho)$ of value
\begin{align*}
    \NSW(\by, \rho) 
    &= \left(\prod_{(a_1,\ldots,a_k) \in \cP_\tau}  \varphi(a_1,\ldots,a_k)\right)^{-\frac1n}\NSW(\by, \pi') \\
 &\geq \left(\prod_{(a_1,\ldots,a_k) \in \cP_\tau} d^k\right)^{-\frac1n}\NSW(\by, \pi')
\end{align*}  
where $\cP_\tau$ is the set of $\tau$-favorable alternating paths. Since the alternating paths are disjoint in terms of the agents they cover, $\prod_{(a_1,\ldots,a_k) \in \cP_\tau} d^{k} \leq d^n$, and together with (\ref{eq:pi'}) we obtain that 
$$ \NSW(\by, \rho) \geq \frac{1}{d} \NSW(\by, \pi') \geq \frac{d-1}{d^{2}} \NSW(\by, \pi)
 \geq \frac{1}{d+2} \NSW(\by,\pi). $$
\end{enumerate}

Now we turn to the guarantee for each agent $a \in \cA$.
If $\rho(a) = \tau(a)$ (i.e. the agent receives an item from the initial matching), then we have either $v_a(\rho(a)) \geq \frac{1}{d} V_a(\by_a)$ which satisfies (i), or by the optimality of the initial matching, we have for every $j \in \cG \setminus \cH$, $v_a(j) \leq v_a(\tau(a)) = v_a(\rho(a)) < \frac{1}{d} V_a(\by_a)$ which satisfies (ii).

If $\rho(a) = \pi(a)$, then this means that $a$ was on a $\pi$-favorable alternating path, and also $a \notin \cB$ because otherwise we would have set $\rho(a) = \pi'(a) = \emptyset$. So this means that $v_a(\rho(a)) = v_a(\pi(a)) \geq \frac{1}{d-1}{V_a(\by_a)}$ by the definition of $\cB$. So we satisfy (i).

The last case is that $\rho(a) = \emptyset$. This means that $a = a_1 \in \cB$ is the starting point of a $\pi$-favorable path $P$, and $\rho(a) = \pi'(a) = \emptyset$. Consider any item $j \in \cG \setminus \cH$. In the initial matching $\tau$, we could replace the $\tau$-edges on $P$ by the $\pi$-edges, and in addition assign $j$ to agent $a_1$. However, this would not result in an improvement since $\tau$ was optimal (as a stand-alone matching). Therefore, we have the following inequality:
$$ \frac{v_{a_1}(j)}{v_{a_1}(\tau(a_1))} \cdot \prod_{j=2}^{k} \frac{v_{a_j}(\pi(a_j))}{v_{a_j}(\tau(a_j))} \leq 1.$$
Recall that $a_2,\ldots,a_k \notin \cB$ and therefore $v_{a_j}(\pi(a_j))\geq \frac{1}{d-1} V_{a_j}(\by_{a_j})$ for $j=2,\ldots,k$.
This implies that $\frac{v_{a_j}(\pi(a_j))}{v_{a_j}(\tau(a_j))} 
\geq \frac{1}{d} \frac{V_{a_j}(\by_{a_j} + \b1_{\pi(a_j)})}{v_{a_j}(\tau(a_j))}
\geq \frac{1}{d} \frac{V_{a_j}(\by_{a_j} + \b1_{\pi(a_j)})}{V_{a_j}(\by_{a_j} + \b1_{\tau(a_j)})}$, 
and also obviously $\frac{v_{a_1}(j)}{v_{a_1}(\tau(a_1))} \geq \frac{v_{a_1}(j)}{V_{a_1}(\by_{a_1} + \b1_{\tau(a_1)})}$.
Therefore, we have
$$ \frac{v_{a_1}(j)}{V_{a_1}(\by_{a_1} + \b1_{\tau(a_1)})} \cdot \prod_{j=2}^{k} \frac{V_{a_j}(\by_{a_j} + \b1_{\pi(a_j)}) }{V_{a_j}(\by_{a_j} + \b1_{\tau(a_j)})} \leq d^{k-1}.$$
Finally, since the path is $\pi$-favorable, we have
$$\varphi(a_1,i_1,\ldots,a_k,i_k) = \frac{V_{a_1}(\by_{a_1})}{V_{a_1}(\by_{a_1} + \b1_{\tau(a_1)})} 
 \prod_{j=2}^{k} \frac{V_{a_j}(\by_{a_j} + \b1_{\pi(a_j)})}{V_{a_j}(\by_{a_j} + \b1_{\tau(a_j)})} > d^k.$$
Combining the last two inequalities, we obtain
$$ v_{a_1}(j) < \frac{1}{d} V_{a_1}(\by_{a_1}) $$
which means that agent $a=a_1$ satisfies (ii).
\end{proof}

\subsection{Conclusion of the analysis}

We conclude the analysis by showing that the matching $\rho$ we proved to exist in Section~\ref{sec:rematching} provides a good value with our fractional solution, even if we ignore the contribution of large items. Hence we can obtain an integral assignment which provides a constant-factor approximation relative to $OPT$ and thus prove Theorem~\ref{thm:main}. 

\begin{lemma}
\label{lem:large-items}
Let $(S_1,\ldots,S_n)$ be the assignment obtained by ${\bf RestrictedRandomizedRounding}(\by)$ with parameter $c>0$ and $\by_i^{(s)} = \by^{(L_i)}_i + \b1_{S_i}$ the sparsified fractional solution. Then there exists a matching $\rho:\cA \rightarrow \cH$ such that
$$ \left( \prod_{i \in \cA} v_i(S_i + \rho(i)) \right)^{1/n} \ge \frac{OPT}{(7+12/c)(c+3)(c+4)}. $$
\end{lemma}

\begin{proof}
Given the sparsified solution $\by^{(s)}$ and the matching $\pi$ provided by Corollary~\ref{cor:rematching}, satisfying
$$ \NSW(\by^{(s)},\pi) \geq \frac{1}{7+12/c} OPT, $$
let $\rho$ be the matching provided by Lemma~\ref{lem:matching} with parameter $d = c+2$. This matching satisfies
$$ \NSW(\by^{(s)},\rho) \geq \frac{1}{c+4} \NSW(\by^{(s)}, \pi) \geq \frac{1}{(7+12/c)(c+4)} OPT $$
and for every agent $i \in \cA$, either (i) $v_i(\rho(i)) \geq \frac{1}{c+2} V_i(\by_i^{(s)})$  or (ii) for every item $j \in \cG' = \cG\backslash\cH$, $v_i(j) < \frac{1}{c+2} V_i(\by_i^{(s)})$. 

For every agent $i\in\cA$, if (i) is the case, then we know that 
$$v_i(S_i + \rho(i))\ge v_i(\rho(i)) \ge \frac{1}{c+3}\left(V_i(\by^{(s)}_i)+v_i(\rho(i))\right)\ge\frac{V_i(\by^{(s)}_i+\vect{1}_{\rho(i)})}{c+3}.$$ 
Otherwise in case (ii), we have 
\begin{align*}
    v_i(S_i + \rho(i))&\ge V_i(\by^{(s)}_i+\vect{1}_{\rho(i)})-\sum_{j\in L_i} y_{ij} v_i(j) \\
    &\ge V_i(\by^{(s)}_i+\vect{1}_{\rho(i)}) - (c+1) \frac{V_i(\by_i^{(s)})}{c+2}  \\ &\ge\frac{V_i(\by^{(s)}_i+\vect{1}_{\rho(i)})}{c+2},
\end{align*}
where the second inequality holds because $\sum_{j\in L_i} y_{ij} \leq c+1$.
In conclusion, we know that for any agent $i\in\cA$, 
$$v_i(S_i + \rho(i)) \ge \frac{1}{c+3} V_i(\by^{(s)}_i+\vect{1}_{\rho(i)}).$$ 
Therefore, we obtain
$$ \left( \prod_{i \in \cA} v_i(S_i + \rho(i)) \right)^{1/n} \ge 
\frac{1}{c+3} \NSW(\by^{(s)}, \rho) \geq \frac{OPT}{(7+12/c)(c+3)(c+4)}.$$
\end{proof}

Now we can prove the main theorem.

\begin{proof}[Proof of Theorem~\ref{thm:main}.]

By Lemma~\ref{lem:large-items}, there is a matching $\rho:\cA \rightarrow \cH$ such that even if we count only the contribution of the small items $S_i$ allocated in ${\bf RestrictedRandomizedRounding}(\by)$,  we have
$$ \left( \prod_{i \in \cA} v_i(S_i + \rho(i)) \right)^{1/n} \ge \frac{OPT}{(7+12/c)(c+3)(c+4)}.$$
This means that the same holds for the sets $R_i$ allocated in our algorithm by ${\bf RandomizedRounding}(\by)$, since $S_i \subseteq R_i$ (the sets $R_i$ include additionally the large items after rounding). In the final step, we find a matching $\sigma$ which is at least as good as $\rho$. 
We choose $c = 1$ which gives
$$ \left( \prod_{i \in \cA} v_i(R_i + \sigma(i)) \right)^{1/n} \ge \left( \prod_{i \in \cA} v_i(S_i + \rho(i)) \right)^{1/n}
\ge \frac{OPT}{19 \cdot 4 \cdot 5} = \frac{OPT}{380}.$$
\end{proof}

\section{Conclusion} 

We have shown a constant-factor approximation algorithm for Nash Social Welfare with submodular valuations, which is the largest natural class of valuations that allows a constant-factor approximation (using value queries) even for additive welfare maximization. However, there are still several directions and open problems to explore. An obvious one is to improve the approximation ratio which is rather large. As we mentioned, we believe that a substantially smaller (say double-digit) factor is hard to achieve with our approach.  

Another open problem is the asymmetric Nash Social Welfare problem, where the objective function is a weighted geometric mean of the agents' valuation functions: $\prod_{i=1}^{n} (v_i(S_i))^{\omega_i}$ for some $\omega_i \geq 0$ (the problem we consider is $\omega_i = 1/n$). The goal is to get a constant-factor approximation independent of the weights $\omega_i$. We remark that \cite{garg2020approximating} gives an approximation guarantee dependent on the weights $\omega_i$; we do not pursue this direction here. 
For the asymmetric problem, getting a universal constant factor is open even in the  the basic case of additive valuations.

Last but not least, solutions optimizing Nash Social Welfare often have additional fairness properties like the envy-free property, or envy-freeness up to one good (see \cite{amanatidis2017approximation}). A line of work has been developed in trying to achieve approximation guarantees for Nash Social Welfare and certain fairness guarantees at the same time \cite{barman2018finding, chaudhury2021little, chaudhury2021fair}. However, our solution does not seem to have such properties and constant-factor approximations with additional fairness guarantees are still unknown for valuation classes beyond additive ones. 
\bibliographystyle{plain}
\bibliography{bibfile}

\begin{thebibliography}{10}

\bibitem{amanatidis2017approximation}
Georgios Amanatidis, Evangelos Markakis, Afshin Nikzad, and Amin Saberi.
\newblock Approximation algorithms for computing maximin share allocations.
\newblock {\em {ACM} Trans. Algorithms}, 13(4):52:1--52:28, 2017.

\bibitem{AMGV18}
Nima Anari, Tung Mai, Shayan~Oveis Gharan, and Vijay~V. Vazirani.
\newblock Nash social welfare for indivisible items under separable,
  piecewise-linear concave utilities.
\newblock In {\em Proceedings of the Twenty-Ninth Annual {ACM-SIAM} Symposium
  on Discrete Algorithms, {SODA} 2018}, pages 2274--2290, 2018.

\bibitem{AGSS17}
Nima Anari, Shayan Oveis~Gharan, Amin Saberi, and Mohit Singh.
\newblock Nash social welfare, matrix permanent, and stable polynomials.
\newblock In {\em 8th Innovations in Theoretical Computer Science Conference
  (ITCS 2017)}. Schloss Dagstuhl-Leibniz-Zentrum fuer Informatik, 2017.

\bibitem{BBKS20}
Siddharth Barman, Umang Bhaskar, Anand Krishna, and Ranjani~G. Sundaram.
\newblock Tight approximation algorithms for p-mean welfare under subadditive
  valuations.
\newblock In {\em 28th Annual European Symposium on Algorithms, {ESA} 2020},
  volume 173 of {\em LIPIcs}, pages 11:1--11:17. Schloss Dagstuhl -
  Leibniz-Zentrum f{\"{u}}r Informatik, 2020.

\bibitem{barman2018finding}
Siddharth Barman, Sanath~Kumar Krishnamurthy, and Rohit Vaish.
\newblock Finding fair and efficient allocations.
\newblock In {\em Proceedings of the 2018 ACM Conference on Economics and
  Computation}, pages 557--574, 2018.

\bibitem{CCPV11}
Gruia Calinescu, Chandra Chekuri, Martin P\'al, and Jan Vondr\'ak.
\newblock Maximizing a monotone submodular function subject to a matroid
  constraint.
\newblock {\em SIAM Journal on Computing}, 40(6):1740--1766, 2011.

\bibitem{CCG18}
Bhaskar~Ray Chaudhury, Yun~Kuen Cheung, Jugal Garg, Naveen Garg, Martin Hoefer,
  and Kurt Mehlhorn.
\newblock On fair division for indivisible items.
\newblock In {\em 38th IARCS Annual Conference on Foundations of Software
  Technology and Theoretical Computer Science (FSTTCS 2018)}. Schloss
  Dagstuhl-Leibniz-Zentrum fuer Informatik, 2018.

\bibitem{chaudhury2021fair}
Bhaskar~Ray Chaudhury, Jugal Garg, and Ruta Mehta.
\newblock Fair and efficient allocations under subadditive valuations.
\newblock In {\em Proceedings of the AAAI Conference on Artificial
  Intelligence}, volume~35, pages 5269--5276, 2021.

\bibitem{chaudhury2021little}
Bhaskar~Ray Chaudhury, Telikepalli Kavitha, Kurt Mehlhorn, and Alkmini
  Sgouritsa.
\newblock A little charity guarantees almost envy-freeness.
\newblock {\em SIAM Journal on Computing}, 50(4):1336--1358, 2021.

\bibitem{cole2017convex}
Richard Cole, Nikhil Devanur, Vasilis Gkatzelis, Kamal Jain, Tung Mai, Vijay~V
  Vazirani, and Sadra Yazdanbod.
\newblock Convex program duality, {F}isher markets, and {N}ash social welfare.
\newblock In {\em Proceedings of the 2017 ACM Conference on Economics and
  Computation}, pages 459--460, 2017.

\bibitem{CG18}
Richard Cole and Vasilis Gkatzelis.
\newblock Approximating the {N}ash social welfare with indivisible items.
\newblock {\em SIAM Journal on Computing}, 47(3):1211--1236, 2018.

\bibitem{DNS10}
Shahar Dobzinski, Noam Nisan, and Michael Schapira.
\newblock Approximation algorithms for combinatorial auctions with
  complement-free bidders.
\newblock {\em Math. Oper. Res.}, 35(1):1--13, 2010.

\bibitem{FNW78}
M.~L. Fisher, G.~L. Nemhauser, and L.~A. Wolsey.
\newblock An analysis of approximations for maximizing submodular set functions
  -- {II}.
\newblock {\em Mathematical Programming Study}, 8:73--87, 1978.

\bibitem{GHM19}
Jugal Garg, Martin Hoefer, and Kurt Mehlhorn.
\newblock Approximating the {N}ash social welfare with budget-additive
  valuations.
\newblock In {\em Proceedings of the Twenty-Ninth Annual ACM-SIAM Symposium on
  Discrete Algorithms}, pages 2326--2340. SIAM, 2018.

\bibitem{garg2020approximating}
Jugal Garg, Edin Husi\'{c}, and L{\'{a}}szl{\'{o}}~A. V{\'{e}}gh.
\newblock Approximating {N}ash social welfare under {R}ado valuations.
\newblock In {\em {STOC} '21: 53rd Annual {ACM} {SIGACT} Symposium on Theory of
  Computing, Virtual Event, Italy, June 21-25, 2021}, pages 1412--1425. {ACM},
  2021.

\bibitem{GKK20}
Jugal Garg, Pooja Kulkarni, and Rucha Kulkarni.
\newblock Approximating {N}ash social welfare under submodular valuations
  through (un) matchings.
\newblock In {\em Proceedings of the fourteenth annual ACM-SIAM symposium on
  discrete algorithms}, pages 2673--2687. SIAM, 2020.

\bibitem{Kel97}
Frank Kelly.
\newblock Charging and rate control for elastic traffic.
\newblock {\em Eur. Trans. Telecommun.}, 8(1):33--37, 1997.

\bibitem{Lee17}
Euiwoong Lee.
\newblock {APX}-hardness of maximizing {N}ash social welfare with indivisible
  items.
\newblock {\em Information Processing Letters}, 122:17--20, 2017.

\bibitem{li2021estimating}
Wenzheng Li and Jan Vondr{\'a}k.
\newblock Estimating the {N}ash social welfare for coverage and other
  submodular valuations.
\newblock In {\em Proceedings of the 2021 ACM-SIAM Symposium on Discrete
  Algorithms (SODA)}, pages 1119--1130. SIAM, 2021.

\bibitem{Nash50}
J.~Nash.
\newblock The bargaining problem.
\newblock {\em Econometrica}, 18(2):155--162, April 1950.

\bibitem{Var74}
H.R. Varian.
\newblock Equality, envy, efficiency.
\newblock {\em Journal of Economic Theory}, 9(1):63--91, 1974.

\bibitem{Von08}
Jan Vondr\'ak.
\newblock Optimal approximation for the submodular welfare problem in the value
  oracle model.
\newblock In {\em Proceedings of the Annual ACM Symposium on Theory of
  Computing}, pages 67--74, 2008.

\end{thebibliography}


\label{LastPage}
\end{document}